\documentclass[11pt]{article}
\input{mysty.sty}
\usepackage{natbib}
\usepackage[showdeletions]{color-edits}

\addauthor{mi}{orange}
\addauthor{jh}{purple}
\newcommand{\cost}{\mathrm{cost}}

\newcommand{\twonorm}[1]{\|#1\|}
\newcommand{\M}{\mathcal{M}}

\title{
Maximum-Cost Strategic Facility Location: \\The Limits of Randomization}
\author{
Jabari Hastings\thanks{Supported by the Simons Foundation Collaboration on the Theory of Algorithmic Fairness and the Simons Foundation Investigators Award 17351.}\\
Stanford University
\and
Misha Ivkov\thanks{This work was partially completed while the author was at Stanford University and supported by the NSF GRFP.}
}

\date{\vspace{15pt}\small{\today}}

\begin{document}

\maketitle

\begin{abstract}
We consider strategic facility location in Euclidean space $\R^d$, where
a mechanism selects a single facility based on the reported locations of
$n$ agents and seeks to minimize the maximum distance from any agent to
the facility.  The optimal approximation ratio of deterministic
strategyproof mechanisms is $2$.  Whether randomization can yield a
universal constant improvement over this factor has remained a major
open question.  We show that for $d \geq 2$, every strategyproof-in-expectation mechanism has
approximation ratio at least
\[
    \alpha(\calM) \ge 2 - e^{-\Theta(\sqrt d)} - O\left(n^{-2/(d-1)}\right).
\]
Therefore, no strategyproof-in-expectation mechanism can guarantee a $(2-\varepsilon)$-approximation uniformly over all $n$ and $d$, for any universal constant $\varepsilon>0$.
\end{abstract}

\clearpage 

\section{Introduction}

Consider the following problem: 

\begin{quote}
  There are $n$ agents in a Euclidean space with distances defined by $\| \cdot \|_2$. Given a profile of their \textit{reported} locations $x_1, x_2, \dots, x_n \in \mathbb R^d$, place a single facility at a location that minimizes the maximum distance between the agents' \emph{actual} locations and the facility.  
\end{quote}
This is a classic instance of the facility location problem, which is widely studied in the mechanism design literature; see \citep{Chan:2021aa} for a survey.
An intriguing aspect of this problem is its strategic dimension.
If the reported locations coincided with the agents' actual locations, then the optimal facility location is not hard to find; it is simply the center of the minimum enclosing ball. However, if an agent knew that the facility location is computed in this way, then she could misreport her location so that the facility is placed at her actual location. The goal of the \emph{strategic} facility location problem is thus to select a facility location that minimizes the maximum distance, subject to a \emph{strategyproofness} constraint: each agent must be incentivized to report her actual location. 

The seminal work of \citet{Procaccia:2009aa,Procaccia:2013aa} demonstrated that strategyproof mechanisms for facility location exist and do not require monetary payments to incentivize agents. 
For facility location in $\R$, they proved that no deterministic strategyproof
mechanism can achieve an approximation ratio better than $2$, whereas
the optimal strategyproof-in-expectation mechanism achieves a ratio of $3/2$. 
Besides establishing the paradigm of mechanism design without money, their results also illustrated that randomization can lead to significantly improved approximations.

For more general network metrics, \citet{Alon:2010aa} obtained a tight
$3/2$-approximation for randomized strategyproof mechanisms on cycles
and showed that no randomized strategyproof mechanism can achieve better
than a $2-o(1)$ approximation on trees under the maximum-cost objective.
In Euclidean spaces, however, the disparity between deterministic and
randomized mechanisms has been less apparent. In $\R^2$, several
deterministic mechanisms yield $2$-approximations, including
dictatorships \citep{Alon:2010aa} and the coordinatewise median
\citep{Goel:2023aa,Chan:2026aa,Hastings:2026aa}; dictatorships remain
$2$-approximations in higher dimensions.  
The only known randomized mechanism that improves upon the factor of $2$
for maximum cost is the \emph{Centroid Mechanism} of
\citet{Tang:2020aa}, which mixes the mean reported location with random
dictatorship and achieves a $(2-\frac{1}{n})$-approximation.  
Whether randomization can yield a constant improvement over the factor
of $2$ in Euclidean spaces has therefore remained a major open question.

Recent work has begun to address this question through lower bounds.
\citet{Balkanski:2024aa} first showed that every
strategyproof-in-expectation mechanism has approximation ratio at least
$\frac{\sqrt{5}}{2}\approx1.118$ in $\R^d$ for $d\geq2$.\footnote{While \cite{Balkanski:2024aa} show this result for $\mathbb R^2$, their argument extends directly to $\mathbb R^d$ for every $d \geq 2$.} Very recently, \citet{Gomes:2026aa} improved the lower bound to $1 + \sqrt {\frac{d}{2(d + 1)} }$, which tends to $1 + \frac{1}{\sqrt{2}} \approx 1.707$ as $d \to \infty$.
Our main result is the following.
\begin{theorem}
    \torestate{
        \label{thm:main-theorem}
        Every strategyproof-in-expectation mechanism $\calM: (\R^d)^n\rightarrow \Delta(\R^d)$ for $d\ge 2$ has approximation ratio
    \[\alpha(\calM) \ge 2 - e^{-\Theta(\sqrt d)} - O\left(n^{-2/(d-1)}\right).\]
    }
\end{theorem}
In particular, taking \(d=\Theta(\log^{2/3} n)\) and balancing the error terms gives
\[
\alpha(\mathcal M)\ge
2-\exp\!\left(-\Omega(\log^{1/3} n)\right).
\]
Therefore, no strategyproof-in-expectation mechanism can guarantee a $(2-\varepsilon)$-approximation uniformly over all $n$ and $d$, for any universal constant $\varepsilon>0$.

\subsection{Related Work}
The facility location problem has been studied from the perspectives of location science,
social choice, and mechanism design; see \citet{Laporte:2020aa},
\citet{Barbera:2011aa}, and \citet{Chan:2021aa}, respectively, for
overviews. Below, we describe the work most closely related to ours.

\paragraph{Strategyproof Mechanisms.}
Strategic facility location can also be viewed as a spatial voting problem
and is closely related to the literature on strategyproof voting over
single-peaked and related preference domains
\citep{Moulin:1980aa,Barbera:1993aa,Ching:1997aa,Schummer:2002aa}.
Most directly relevant, \citet{Kim:1984aa} characterize continuous,
anonymous, strategyproof rules under Euclidean preferences in $\R^2$ as
coordinatewise generalized median mechanisms.  Two companion papers
extend this picture by relating continuity to convexity of the range and
by studying Pareto-optimal, anonymous mechanisms under general strictly
convex norms \citep{Peters:1993aa,Peters:1993ab}.
Related characterizations under other assumptions on the range and
preference domain appear in
\citet{Border:1983aa,Barbera:1998aa,Peremans:1997aa}.
These results identify important classes of deterministic mechanisms, but
do not yield a characterization at the level of generality considered in
this work.

\paragraph{Randomized Mechanisms for Other Objectives.}
Recent work has studied randomized strategyproof mechanisms under other
social-cost objectives.
\citet{Barak:2026aa} studies strategic facility location with utilitarian social cost. For facility location in $\mathbb R^2$, the author introduces the \emph{Randomly Rotated Coordinatewise Median} (RR-CWM), a randomized \emph{universally strategyproof} mechanism with an approximation ratio of $\frac{4}{\pi} \approx 1.273$.
This provides a strict improvement over the optimal deterministic strategyproof mechanism, which is a $\sqrt{2}$-approximation. Establishing a strict separation in higher dimensions is an open problem. \citet{Barak:2026aa} also establishes lower bounds for a class of \emph{Generalized Random Dictators}. However, there are no known general randomized lower bounds for the utilitarian setting.

Independently, \citet{Chan:2026aa} analyze the same rotated mechanism
under $\ell_p$-norm social costs in $\R^2$.  They show that RR-CWM has a strictly better approximation ratio than the coordinatewise median for $1\leq p<2$, while the
\emph{Centroid Mechanism} improves upon RR-CWM for finite
$p\gtrsim1.6$.  General randomized lower bounds remain unknown for
finite $p$.

\paragraph{Multiple Facilities.}
Beyond the single facility setting, \citet{Procaccia:2013aa}
also studied the two-facility variant in $\mathbb R$.
For deterministic mechanisms, they established an $\Omega(n)$-approximation ratio for the utilitarian social cost and a $2$-approximation ratio for the maximum cost. For strategyproof-in-expectation mechanisms for the maximum social cost, they established a lower bound of $3/2$ and an upper bound of $5/3$.

For the two-facility problem in \emph{general metric spaces} with utilitarian social cost, \citet{Lu:2010aa} proved that every deterministic strategyproof mechanism has approximation ratio $\Omega(n)$ and gave a strategyproof-in-expectation mechanism with approximation ratio $4$.
Very recently for the same setting, \citet{Ma:2026aa} improved the upper bound to $11/3 \approx 3.667$. They also strengthened the lower bound of \cite{Lu:2009aa} from $1.045$ to $(1 + \sqrt{2})/2 \approx 1.207$. Their lower bound relies on analyzing the effect of manipulating a single agent within a collocated block. Interestingly, our main result also works with collocated blocks, but in a different way: we consider coordinated perturbations along orthogonal directions in $\mathbb R^d$.

\section{Preliminaries}
In the strategic facility location problem, the true locations of $n$ agents are denoted by the location profile $x = (x_1, \dots, x_n ) \in (\R^d)^n$. The space $\R^d$ is equipped with the $\ell_2$-norm, denoted by $\twonorm{\cdot}$. The individual \emph{cost} of a facility location $f \in \R^d$ for an agent located at $x_i$ is the distance $\twonorm{f - x_i}$. 

\paragraph{Social Cost.}
The \emph{egalitarian social cost} of a facility $f$ aggregates individual costs using the $\ell_\infty$-norm,
\begin{equation}
    \cost(f, x) =  \max_{i \in [n]} \twonorm{f - x_i}.
\end{equation}
A \textit{randomized} mechanism $\M$ takes the reported location profile $x$ as input and produces a probability distribution over $\R^d$. We slightly abuse notation to define the expected social cost of the mechanism on profile $x$ as
\begin{equation}
    \cost(\M, x) = \E_{F \sim \M(x)}[\cost(F, x)],
\end{equation}
where the facility $F$ is drawn from the distribution $\M(x)$.

\paragraph{Approximation Ratio.}
The \textit{approximation ratio} of a mechanism $\M$ is the worst-case ratio of the mechanism's expected social cost to the optimal social cost over all location profiles $x \in (\R^d)^n$,
\begin{equation}
    \alpha(\M) = \sup_{x \in (\R^{d})^n} \frac{\cost(\M, x)}{\min_{f \in \R^d} \cost(f, x)}.
\end{equation}

Note that when all agents are collocated, $\cost(\calM, x)$ must be $0$ to avoid an infinite approximation ratio; we will disregard this exceptional case from the definition of $\alpha(\calM)$.

\paragraph{Strategyproofness.}
A randomized mechanism $\M$ is \emph{strategyproof-in-expectation} if no agent can decrease their expected cost by reporting a false location. Formally, for any agent $i \in [n]$ and any two location profiles $x, x' \in (\R^d)^n$ that differ only in the $i$-th coordinate (i.e., $x_j = x_j'$ for all $j \neq i$),  
\begin{equation*}
    \EE_{F \sim \M(x)}[\twonorm{F - x_i}] \le \EE_{F \sim \M(x')}[\twonorm{F - x_i}].
\end{equation*}

A useful consequence of strategyproofness-in-expectation is the following surprising fact regarding \emph{group deviation} noted in the recent work of \citet{Gomes:2026aa}.

\begin{fact}[Lemma 3.3 of~\cite{Gomes:2026aa}]
\label{fact:group-deviation}
    Suppose that $x$ is a profile where $k$ agents are co-located at a point $p$.
    Then, for every location profile $x'$ which only differs from $x$ on these $k$ agents,
    \[\EE_{F\sim \calM(x')}[\|F - p\|] \ge \EE_{F\sim \calM(x)}[\|F - p\|].\]
\end{fact}

\section{Proof Overview}

Before giving the main proof, it is insightful to instead aim to prove that any mechanism $\cal M$ that is  strategyproof-in-expectation must have approximation ratio $\alpha(\calM) \ge 2 - O\left(\frac 1{\sqrt d}\right) - o_n(1)$.
For this discussion, we will assume that $n$ is sufficiently large, even compared to $d$.

The key insight we use is that the construction of an adversarial agent profile for the mechanism $\cal M$ need not be one-shot.
In particular, we can adaptively construct the profile in \emph{stages}, choosing configurations which make $\calM$ perform poorly. For each such configuration, the mechanism must ensure that no agent benefits from misreporting at the expense of minimizing the maximum cost.

Toward this adaptive construction of an agent profile, it is useful to consider the function $\calO_{\calM}(p;x) \coloneqq \E_{F\sim \calM(x)}[\|F - p\|]$.
Note that if there is at least one agent at $p$ in $x$, then $\calO_{\calM}(p;x)$ is a lower bound on $\cost(\calM, x)$.
Thus, for any profile $x$, it follows that
\begin{align}\alpha(\calM) \ge \frac{\max_{p\in x}\calO_{\calM}(p;x)}{\min_{f\in \R^d}\cost(f, x)}\label{eq:lossy-bound}.\end{align}
The construction may be encoded by the following algorithm; $T$ is the aforementioned number of stages and $s$ is a branching parameter controlling how much $\calO_{\calM}$ increases at each stage.

\begin{algorithmSELF}[Construction of agent profile]
\label{alg:agent-profile-construction}
{\textbf{Input: }} Oracle $\calO_{\calM}$; branching parameter $s$;  stage parameter $T \geq 2$ with $d\ge 1 + Ts$ and $s \ge T$.

\noindent {\textbf{Operation: }}
\begin{itemize}[noitemsep]
    \item Choose an orthonormal basis $u_1, u_2, \ldots, u_d$ for $\R^d$ and partition it into chunks 
    \[[\calB_0,\calB_1,\ldots]= [[u_1], [u_2,\ldots,u_{s+1}], [u_{s+2},\ldots,u_{2s+1}],\ldots]\]
    \item Initialize the profile $x^{(0)}$ to have $\frac n2$ agents at $u_1$ and $\frac n2$ agents at $-u_1$, and suppose WLOG that $\calO_{\calM}(u_1; x^{(0)}) \ge \calO_{\calM}(-u_1; x^{(0)})$.
    \item Initialize the current position $p = u_1$.
    \item {\textbf{For}} $1\le i\le T$:
    \begin{itemize}
        \item [$\spadesuit$] {\color{blue} \textbf{Split} the agents at $p$ in profile $x^{(i - 1)}$ into $2s$ blocks of equal size, and move them to $\{p \pm u_j: u_j \in \calB_i\}$ to form profile $x^{(i)}$.}
        \item Among the $2s$  newly created blocks, choose the point $p'$  maximizing $\calO_{\calM}(p'; x^{(i)})$, and set $p = p'$.
    \end{itemize}
    \item {\textbf{Return}} $x^{(T)}$.
    
\end{itemize}

\noindent {\textbf{Output: }} An agent profile $x^{(T)} \in (\R^d)^n$.
\end{algorithmSELF}

To understand ~\pref{alg:agent-profile-construction}, it is helpful to consider the toy example $d = 3, s = 1, T = 2$ (even though it lies outside the $s \geq T$ regime); see~\pref{fig:evolution-of-agents}.
In this example, we take the orthonormal basis $e_1, e_2, e_3$ (in other words, the coordinate axes), and thus $[\calB_0, \calB_1, \calB_2] = [[e_1], [e_2], [e_3]]$.

\tdplotsetmaincoords{70}{120}

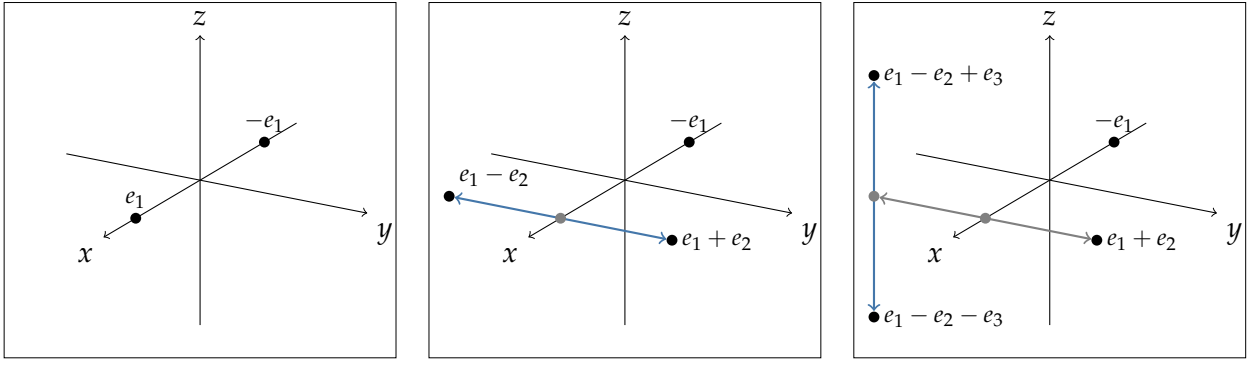
\begin{figure}[ht]

\centering
\begin{subfigure}[b]{0.32\textwidth}
    \centering
    \begin{tikzpicture}[tdplot_main_coords, scale=1.7]

        \node[
            draw,
            line width=0.2pt,
            minimum width=0.98\linewidth,
            minimum height=4.7cm,
            inner sep=0pt
        ] at (0,0,0) {};

        \draw[->] (-1.5,0,0) -- (1.5,0,0) node[anchor=north east] {$x$};
        \draw[->] (0,-1.2,0) -- (0,1.5,0) node[anchor=north west] {$y$};
        \draw[->] (0,0,-1.2) -- (0,0,1.2) node[anchor=south] {$z$};
    
        \fill (1,0,0) circle (1.2pt);
        \fill (-1,0,0) circle (1.2pt);
    
        \node[above] at (1,0,0) {\footnotesize $e_1$};
        \node[above] at (-1,0,0) {\footnotesize $-e_1$};
    \end{tikzpicture}
    \caption{$x^{(0)}$: agents are split between $e_1$ and $-e_1$.}
\end{subfigure}%
\hfill
\begin{subfigure}[b]{0.32\textwidth}
    \centering
    \begin{tikzpicture}[tdplot_main_coords, scale=1.7]

        \node[
            draw,
            line width=0.2pt,
            minimum width=0.98\linewidth,
            minimum height=4.7cm,
            inner sep=0pt
        ] at (0,0,0) {};

        \draw[->] (-1.5,0,0) -- (1.5,0,0) node[anchor=north east] {$x$};
        \draw[->] (0,-1.2,0) -- (0,1.5,0) node[anchor=north west] {$y$};
        \draw[->] (0,0,-1.2) -- (0,0,1.2) node[anchor=south] {$z$};
    
        \draw[->, thick, blue] (1,0,0) -- (1,0.95,0);
        \draw[->, thick, blue] (1,0,0) -- (1,-0.95,0);
    
        \fill[gray] (1,0,0) circle (1.2pt);
        \fill (-1,0,0) circle (1.2pt);
        \fill (1,1,0) circle (1.2pt);
        \fill (1,-1,0) circle (1.2pt);
    
        \node[above] at (-1,0,0) {\footnotesize $-e_1$};
        \node[right] at (1,1,0) {\footnotesize $e_1+e_2$};
        \node[above right] at (1,-1,0) {\footnotesize $e_1-e_2$};
    \end{tikzpicture}
    \caption{$x^{(1)}$: agents at $p = e_1$ are split between $e_1\pm e_2$.}
\end{subfigure}%
\hfill
\begin{subfigure}[b]{0.32\textwidth}
    \centering
    \begin{tikzpicture}[tdplot_main_coords, scale=1.7]

        \node[
            draw,
            line width=0.2pt,
            minimum width=0.98\linewidth,
            minimum height=4.7cm,
            inner sep=0pt
        ] at (0,0,0) {};

        \draw[->] (-1.5,0,0) -- (1.5,0,0) node[anchor=north east] {$x$};
        \draw[->] (0,-1.2,0) -- (0,1.5,0) node[anchor=north west] {$y$};
        \draw[->] (0,0,-1.2) -- (0,0,1.2) node[anchor=south] {$z$};
    
        \draw[->, thick, gray] (1,0,0) -- (1,0.95,0);
        \draw[->, thick, gray] (1,0,0) -- (1,-0.95,0);
        \draw[->, thick, blue] (1,-1,0) -- (1,-1,0.95);
        \draw[->, thick, blue] (1,-1,0) -- (1,-1,-0.95);
    
        \fill[gray] (1,0,0) circle (1.2pt);
        \fill (-1,0,0) circle (1.2pt);
        \fill (1,1,0) circle (1.2pt);
        \fill[gray] (1,-1,0) circle (1.2pt);
        \fill (1,-1,1) circle (1.2pt);
        \fill (1,-1,-1) circle (1.2pt);
    
        \node[above] at (-1,0,0) {\footnotesize $-e_1$};
        \node[right] at (1,1,0) {\footnotesize $e_1+e_2$};
        \node[right] at (1,-1,1) {\footnotesize $e_1-e_2+e_3$};
        \node[right] at (1,-1,-1) {\footnotesize $e_1 -e_2 - e_3$};
    \end{tikzpicture}
    \caption{$x^{(2)}$: agents at $p = e_1 - e_2$ are split between $e_1-e_2\pm e_3$.}
\end{subfigure}

\caption{Illustration of the splitting operation for $d=3, s = 1, T = 2$. The final value of $p$ is either $e_1 - e_2 + e_3$ or $e_1 - e_2 - e_3$, depending on the last selection.
\label{fig:evolution-of-agents}
}
\end{figure}

By~\pref{eq:lossy-bound}, it suffices to understand $\calO_{\calM}(p;x^{(T)})$ and $\min_{f\in \R^d} \cost(f, x^{(T)})$. Below we will sketch arguments bounding these.

\begin{claim}
    $\calO_{\calM}(p;x^{(T)}) \ge \sqrt{1 + T(1 - \frac 2s)}$.
\end{claim}
\begin{proof}
The proof is by induction on the stage $i$.
In particular, let the sequence of points $p$ through the $T$ stages be $p_0, p_1, \ldots, p_T$.
By construction, $\calO_{\calM}(p_0;x^{(0)}) \ge 1$.

For the induction step, consider the symmetrization of $\mathcal B_{i+1}$ given by $$\overline {\calB}_{i+1} = \calB_{i+1} \cup \{-u_j : u_j\in \calB_{i+1}\}.$$

We claim the following pointwise inequality for all $z\in \R^d$ assuming $s \ge 2$:
\begin{equation}\E_{v\in \overline \calB_{i+1}}[\|z - v\|_2] \ge \sqrt{\|z\|^2 + \left(1 - \frac 2s\right)}.\footnote{Note that this is almost tight. 
In particular, by Jensen's, we have that $\E[\|z - v\|_2]\le \sqrt{\E[\|z - v\|_2^2]} = \sqrt{\|z\|_2^2 + 1}$, since $\overline \calB_{i+1}$ is a symmetric set of unit vectors.
As $s\rightarrow\infty$, the pointwise inequality converges to this bound.}
\label{eq:pointwise-increase}\end{equation}
This allows us to iterate:
    \begin{align*}\max_{v\in \overline \calB_{i+1}} \calO_{\calM}(p_i + v; x^{(i+1)}) &\ge \EE_{v\in \overline \calB_{i+1}} [\calO_{\calM}(p_i + v; x^{(i+1)})]\\
    &= \EE_{F\sim \calM(x^{(i+1)})}\left[\EE_{v\in \overline \calB_{i+1}} [\|F - (p_i + v)\| ] \right]\\
    &\overset{\pref{eq:pointwise-increase}}{\ge}\EE_{F\sim \calM(x^{(i+1)})}\left[\sqrt{\|F - p_i\|_2^2 + \left(1 - \frac 2s\right)}\right]\\
    &\!\!\!\!\!\!\!\overset{\text{(Jensen's)}}{\ge} \sqrt{\EE_{F\sim \calM(x^{(i+1)})}[\|F - p_i\|_2]^2 + \left(1 - \frac 2s\right)}\\
    &\!\!\!\!\!\overset{\text{(\pref{fact:group-deviation})}}{\ge} \sqrt{\EE_{F\sim \calM(x^{(i)})}[\|F - p_i\|_2]^2 + \left(1 - \frac 2s\right)}\\
    &\ge \sqrt{1 + (i+1)\left(1 - \frac 2s\right)}.
    \end{align*}
and thus the induction is complete.
\end{proof}

\begin{claim}
$\min_{f\in \R^d} \cost(f, x^{(T)}) \le \frac 12\sqrt{T + O(1)}$.
\end{claim}
\begin{proof}[Proof Sketch]
We claim that $\max_{q\in x^{(T)}}\|q - \frac 12 p\| \le \frac 12 \sqrt{T + O(1)}$, which implies the claim.

We start by computing $\|p_i - \frac 12 p\|^2$.
By design (using the orthogonality of all directions), it follows that $\|p_i - \frac 12 p\|^2 = \frac {T+1}{4}$: in each of the $T + 1$ directions added to $p$, the coefficient is $\pm \frac 12$.

Now, suppose that an agent location $q$ comes from the $i$'th stage.
For each $1\le i\le T$, let
$u_i^\star=p_i-p_{i-1}$ denote the direction selected at stage $i$. Then, $q - p_i = v - u_i$ where either $v = -u_i$ or $v$ is orthogonal to all $u_j$.
In either scenario, $\|q - \frac 12 p\|^2 = \|p_i - \frac 12 p\|^2 + O(1)$ since at most two orthogonal directions are changed by $O(1)$.
\end{proof}

\paragraph{Putting it all together.} 
We will take $s = T = \sqrt{d - 1}$. 
Then, we have that
\[\alpha(\calM) \ge 2\frac{\sqrt{T  - 1}}{\sqrt{T + O(1)}} = 2 - O\left(\frac 1T\right) = 2 - O\left(\frac 1{\sqrt d}\right)\]
by a first-order Taylor approximation of $x \mapsto \sqrt{x}$.

\paragraph{Improving to exponentially small dependence on $d$.}

With some cleverness in reusing orthonormal directions within the $T$ stages, it is possible to improve the dependence to $2 - O\left(\frac 1d\right)$ (in effect, we only require $d\ge T + s$); however, the barrier remains that $\calO_{\calM}(p;x^{(i)})^2$ only grows linearly in the number of stages.
As noted in a footnote, this is essentially optimal for~\pref{eq:pointwise-increase} even as $s\rightarrow \infty$, so we cannot hope to beat linear growth.

To reach our improved bound, we first modify the perturbation to not be of unit norm.
This allows us to instead get \emph{multiplicative} growth in $\calO_{\calM}$.

By virtue of our new perturbation, we will also show that all of $x^{(T)}$ is ``close'' to the final point $p_T$, which allows us to use the ending part of the construction of~\cite[Theorem 3.1]{Gomes:2026aa}.

\begin{proposition}[Bad point to large approximation ratio]
\label{prop:large-apx-ratio}
    Let $x$ be a location profile and $p$ a point with at least $2d\left(\frac {3}{\sqrt{\eps}}\right)^{d - 1}$ agents on it such that $\calO_{\calM}(p; x)\ge D$ and all of $x$ lies within $B(p, R)$.
Then, $\alpha(\calM) \ge 1 + \frac DR - \eps$.
\end{proposition}

\begin{proof}
    Take all agents currently at $p$ and spread them on an $R\sqrt{\eps}$-net of $\partial B(p, R)$, creating a new profile $x'$ with radius $R$.
Consider an arbitrary point $f\in \R^d$. We claim that \[\max_{q\in x'}\|f - q\|_2^2 \ge (\|f - p\| + R)^2 - \|f - p\|R\eps.\]
By rescaling and translating, we may assume that $p = 0$ and $R = 1$.
Now, consider the unit antipode $\overline f$ of $f$, as projected onto $\partial B(0, 1)$.
    There is some $q$ in our net such that $\|q - \overline f\|\le \sqrt{\eps}$, so $\langle q,\overline f\rangle \ge 1 - \frac{\eps}{2}$ and $\langle q, f\rangle \le -\|f\|(1 - \frac{\eps}{2})$.

    Then, it follows that 
    \[\|f - q\|^2 = \|f\|^2 + \|q\|^2 - 2\langle f, q\rangle \ge \|f\|^2 + 1 + 2\|f\| - \|f\|\eps = (\|f\| + 1)^2 - \|f\|\eps.\]
    Scaling back and taking a square root, we may then write
    \[\|f - q\|_2 \ge \sqrt{(\|f - p\| + R)^2 - \|f - p\|R\eps} \ge \|f - p\| + R - \frac{\|f - p\|R\eps}{\|f - p\| + R} \ge \|f - p\| + R(1 - \eps)\]
    by using $\sqrt{x - \eps} \ge \sqrt{x} - \frac{\eps}{\sqrt x}$.
    Thus, we have that
    \[\alpha(\calM) \ge \frac{\E_{F\sim \calM(x')}[\max_{q\in x'} \|F - q\|_2]}{R}\ge \frac{\E_{F\sim \calM(x')}[\|F - p\|_2] + R(1 - \eps)}{R} \ge 1 + \frac DR - \eps\]
    as desired.
\end{proof}
\section{Main Result}

In this section, we prove the main result.

\restatetheorem{thm:main-theorem}

The construction is similar to~\pref{alg:agent-profile-construction}. {\bf Throughout, we take $\lambda = 0.9$.} 
\begin{algorithmSELF}[Modified construction algorithm]
\label{alg:new-profile-construction}
Everything remains unchanged, except the line marked $\spadesuit$, which becomes \\

\noindent
 Let $a^\ast = \lambda + \frac {2}{s}$. 
{\textbf{Split}} the agents at $p$ in $x^{(i-1)}$ into $2s$ blocks of equal size, and move them to $\left\{\frac{\lambda}{a^\ast}p \pm \frac{\sqrt{1-\lambda^2}}{(a^\ast)^i}u_j: u_j \in \calB_i\right\}$ to form $x^{(i)}$.
\end{algorithmSELF}

In light of~\pref{prop:large-apx-ratio}, we show the following result, where we will drop all floors and ceilings (in particular, we will assume $n$ is divisible by $2(2s)^T$).

\begin{lemma}[Constructing a bad point]
\label{lem:construct-bad-point}
    Run~\pref{alg:new-profile-construction} and assume that $T\le \frac{s}{25}$ and $s\ge 25$.
At the end of the algorithm, the point $p$ has at least $\frac{n}{2(2s)^T}$ agents and all of $x^{(T)}$ lies within $B(p, R)$ where $R = \calO_{\calM}(p;x^{(T)})\cdot \left(1 + e^{-\Theta(T)}\right)$.
\end{lemma}

Before proving the lemma, we show it implies~\pref{thm:main-theorem}.

\begin{proof}[Proof of Theorem~\ref{thm:main-theorem}]
    Choose $s = 5\sqrt{d - 1}$ and $T = \frac 1{5} \sqrt{d - 1}$ in~\pref{alg:new-profile-construction} and apply~\pref{lem:construct-bad-point}.

    Then, we are guaranteed that $\frac DR \ge 1 - e^{-\Theta(T)} = 1 - e^{-\Theta(\sqrt d)}$, and that there are at least $\frac{n}{2(10\sqrt{d-1})^{\frac 1{5}\sqrt{d-1}}}$ agents at $p$.
    Now, applying~\pref{prop:large-apx-ratio}, the net parameter $\eps$ must satisfy
    \[2d\left(\frac 3{\sqrt{\eps}}\right)^{d-1} \le \frac{n}{2(10\sqrt{d-1})^{\frac 1{5}\sqrt{d-1}}} \implies \frac 3{\sqrt{\eps}} \le \Theta\left(n^{1/(d-1)}\right) \implies \eps \ge \Theta\left(n^{-2/(d-1)}\right).\]
    Setting such a value of $\epsilon$ and taking the lower bound on $D/R$ completes the proof.
\end{proof}

To prove~\pref{lem:construct-bad-point},  we will require the following structural result.

\begin{lemma}[Existence of far directions]
\torestate{
\label{lem:far-directions}
    Suppose that $S\subseteq \bbS^{d-1}$ is an orthonormal set with $|S| = s \ge 25$.
Then, for every $D\in \R$, $z\in \R^d$ and $u\perp S$ with $\|u\|=1$,
    \[\operatornamewithlimits{\E}\limits_{v\in (S)\cup (-S)} [\|z - D(\lambda u + \sqrt{1-\lambda^2} v)\|_2] \ge \sqrt{D^2(1 - \lambda a^\ast) + \frac{\lambda}{a^\ast}\|z - Da^\ast u\|^2}.\]
}
\end{lemma}

We leave the proof of this to~\pref{app:appendix-miscellaneous}.

\begin{proof}[Proof of Lemma~\ref{lem:construct-bad-point}]
    We will first prove that $\calO_{\calM}(p_i;x^{(i)})\ge \|p_i\|$ by induction.
    Note that this is satisfied at $i = 0$, since $\twonorm{p_0} = 1$ and $\calO_{\calM}(p_0;x^{(0)}) \geq 1$.

    As before, let $\overline \calB_{i+1}$ be the symmetrization of $\calB_{i+1}$, and compute $D$ such that $p_i = Da^\ast\overline p$ with $\|\overline p\| = 1$.

    Now, we see that (similar to the proof overview)
    \begin{align*}
        \max_{v\in \overline \calB_{i+1}}\Biggl[&\EE_{F\sim \calM(x^{(i+1)})}\left[\left\|F - D\left(\lambda \overline p + \sqrt{1 - \lambda^2}v\right)\right\|\right]\Biggr]\\
        &\!\!\!\!\!\!\!\!\overset{\text{(\pref{lem:far-directions})}}{\ge} \EE_{F\sim \calM(x^{(i+1)})}\left[\sqrt{D^2(1 - \lambda a^\ast) + \frac{\lambda}{a^\ast}\|F - p_i\|^2}\right]\\
        &\ge \sqrt{D^2(1 - \lambda a^\ast) + \frac{\lambda}{a^\ast}\EE_{F\sim \calM(x^{(i+1)})}[\|F - p_i\|]^2}\\
        &\ge \sqrt{D^2(1 - \lambda a^\ast) + \frac{\lambda}{a^\ast}\EE_{F\sim \calM(x^{(i)})}[\|F - p_i\|]^2}\\
        &\ge \sqrt{D^2(1 - \lambda a^\ast) + D^2\lambda a^\ast}\\
        &= D
    \end{align*}
    by induction, Jensen's inequality, and~\pref{fact:group-deviation}; thus some $v$ achieves this bound.
Since $p_i \perp v$, we know that $\|p_{i+1}\| = D\left\|\lambda \overline p + \sqrt{1 - \lambda^2}v\right\| = D$, which concludes the induction.

Our final goal is to show that $\max_{q\in x^{(T)}} \frac{\|q - p_T\|}{\|p_T\|} \le 1 + e^{-\Theta(T)}$: this will imply that all of $x^{(T)}$ lies within $B(p_T, \|p_T\|(1 + e^{-\Theta(T)}))$ as desired.

To do so, it is useful to characterize what $p_t$ looks like.
In particular, suppose that the orthogonal directions chosen at every stage are $v_0 = e_1, v_1, v_2, \ldots, v_T$. 
Then, we can write
\[p_t = \frac 1{(a^\ast)^t}\left(\lambda^t v_0 + \sqrt{1-\lambda^2}\sum_{i=1}^{t} \lambda^{t-i} v_i\right) \]
and note that $\|p_t\| = (a^\ast)^{-t}$.

We will use the following claim, which we prove in~\pref{app:appendix-miscellaneous}.
\begin{claim}
\torestate{
\label{clm:distance-to-final-point}
    The following equality holds for all $0\le t\le T$.
    \[\frac{\|p_t - p_T\|_2^2}{\|p_T\|^2} = 1 + (a^\ast)^{2(T - t)} - 2(\lambda a^\ast)^{T - t}.\]
}
\end{claim}

This also lets us control $\|q - p_T\|$ for any $q\in x^{(T)}$.
Indeed, every $q \in x^{(T)} \setminus \{-e_1, p_T \}$ can be written in the form $p_t + \frac{\sqrt{1 - \lambda^2}}{(a^\ast)^t}(u - v_t)$, where $u = -v_t$ or $u$ is orthogonal to all $v_i$ (but still unit).
\begin{itemize}
    \item First, if $u$ is orthogonal to $v_t$, then we have to both edit the $v_t$ component and add back in $u$:
\begin{align}
&\frac{\|q - p_T\|_2^2}{\|p_T\|_2^2} \nonumber\\
& = 1 + (a^\ast)^{2(T - t)} - 2(\lambda a^\ast)^{T - t} - (1 - \lambda^2)((a^\ast)^{T - t} - \lambda^{T-t})^2 + (1 - \lambda^2)\left[\lambda^{2(T - t)} + (a^\ast)^{2(T-t)}\right]\nonumber \\
&= 1 + (a^\ast)^{2(T-t)} - 2(\lambda a^\ast)^{T - t} + 2(1 - \lambda^2)(\lambda a^\ast)^{T - t}\nonumber\\
&= 1 + (a^\ast)^{2(T-t)} - 2\lambda^2 (\lambda a^\ast)^{T - t}\label{eq:type-1}.\end{align}
\item Next, suppose that $u = -v_t$. 
Then, we can correct the equality by flipping the sign of the $v_t$ component: this yields
\begin{align}
&\frac{\|q - p_T\|_2^2}{\|p_T\|_2^2}\nonumber \\
&\quad= 1 + (a^\ast)^{2(T - t)} - 2(\lambda a^\ast)^{T - t} - (1-\lambda^2)((a^\ast)^{T-t} - \lambda^{T - t})^2 + (1-\lambda^2)((a^\ast)^{T-t} + \lambda^{T - t})^2\nonumber \\
&\quad= 1 + (a^\ast)^{2(T - t)} - 2(\lambda a^\ast)^{T - t} + 4(1 - \lambda^2)(\lambda a^\ast)^{T-t}\nonumber \\
&\quad= 1 + (a^\ast)^{2(T - t)} - (4\lambda^2 - 2)(\lambda a^\ast)^{T-t}.\label{eq:type-2}
\end{align}

\item Finally, we see that $\frac{\|p_T + e_1\|}{\|p_T\|} \le 1 + (a^\ast)^T = 1 + e^{-\Theta(T)}$ since $a^\ast < 1$.
\end{itemize}

We will show that both of these first two cases are at most $1$, when $T\le \frac{s}{20}$.
To do so, we wish to understand when $\kappa (\lambda a^\ast)^{T-t} \ge (a^\ast)^{2(T-t)}$.
The case $t=T$ is immediate; hence assume $t<T$.
Taking the $(T - t)$th roots of both sides, this requires that $\kappa^{\frac{1}{T - t}} \ge \frac{a^\ast}{\lambda} = 1 + \frac{2}{\lambda s}$.
Certainly, the left hand side is minimized at $t = 0$. 
Noting that $(1 + x)^{1/T} \ge 1 + \frac{x}{2T}$ when $0\le x\le 1$,\footnote{Raising both sides to the $T$'th power, we have $(1 + \frac x{2T})^T \le e^{x/2} \le 1 + x$ for $0\le x\le 1$, which can be seen by noting that the derivative of $1 + x - e^{x/2}$ is positive.} it follows that $\kappa^{\frac 1{T - t}} \ge 1 + \frac{\kappa - 1}{2T}$, and it suffices to have $T\le \frac{\lambda(\kappa - 1) s}{4}$.
When $\lambda = \frac 9{10}$, we may evaluate~\pref{eq:type-1} and~\pref{eq:type-2} to give $\kappa \ge 1.24$ and therefore $T\le \frac{s}{20}$ is indeed sufficient to make the first two cases at most $1$.
\medskip

Putting it all together, we have shown that $\max_{q\in x^{(T)}} \frac{\|q - p_T\|}{\|p_T\|} \le 1 + (a^\ast)^T = 1 + e^{-\Theta(T)}$, and since additionally $\calO_{\calM}(p_T;x^{(T)}) \ge \|p_T\|$ we have proven the lemma.
\end{proof}
\section{Discussion}

\paragraph{Future Work.}
The glaring open question is: is $2 - o_n(1)$ the correct bound for all $d\ge 2$?
Our result only approaches $2$ as $d\rightarrow \infty$, for sufficiently large $n$ compared to $d$; furthermore, it is only an asymptotic result and does not immediately give an improvement for, say, $d = 2$.

From the upper bound direction, we know that the \emph{Centroid Mechanism} \citep{Tang:2020aa} achieves approximation ratio $2 - \frac 1n$; is this optimal in all dimensions?
Even in the $d\rightarrow\infty$ regime, our bound suffers in the number of agents compared to this mechanism (although we have not attempted to optimize constants for the sake of readability).

From a theoretical standpoint, our adaptive construction may also be useful for lower bounds in other regimes, such as in the two-facility setting recently considered by~\citet{Ma:2026aa}. 

\paragraph{AI Disclosure.}
The main result of this paper arose from an interactive session with GPT 5.6 Sol Extra High and Pro in Codex. The model did not yield a one-shot solution; however the authors gradually fed the model concrete ideas based on preliminary results derived by the authors.
Reflecting on our insights, we wrote the following lightly edited prompt:
\begin{quote}
    \textit{Note that taking perturbations from a collocated block of agents appears to be quite useful. 
    Many of our earlier arguments have tried to show some sort of insensitivity. If the total change after a coordinated manipulation is at least 1, then the change along an average direction is roughly $1/n$. Can we use these ideas to prove a $2 - \eps$ lower bound for $n$ agents in $\mathbb R^n$? }
\end{quote}
The idea was to find a direction such that the distance between the mechanism's facility locations for the initial profile and a perturbed profile is large and to exploit this with repeated manipulations.

Early responses to this prompt did not lead to our desired  unconditional lower bounds. However, with subsequent interaction, it occurred to the authors that taking perturbations along orthogonal components in $\mathbb R^d$ may suffice. Eventually, the model generated a lower bound of $2 - O(\frac{1}{d^{1/2}} + \frac{\log \log n}{\log n})$ (an exposition of which can be seen in the proof overview). The authors read and verified this lower bound and then manually adjusted the construction to have inverse linear dependence on $d$. With the appropriate prompt combined with a write-up of this construction, the model further improved the dependence to be as stated in \Cref{thm:main-theorem}.

The content of this paper was written by the human authors. All the key ideas from the interactions with GPT 5.6 Sol were reorganized and communicated by the authors, and parts of the proof were strengthened. The authors accept full responsibility for the correctness of results stated in this paper.

\clearpage 
\bibliographystyle{apalike} %unsrt
\bibliography{references}

\clearpage
\appendix

\section{Miscellaneous Proofs}
\label{app:appendix-miscellaneous}

\restatelemma{lem:far-directions}

\begin{proof}
By symmetry, it suffices to consider \(D\ge0\). The case \(D=0\) is immediate, and for \(D>0\), rescaling by \(D\) reduces to \(D=1\).
We begin by decomposing $z' = z  - \lambda u = \sum_{v\in S} \alpha_v v + z^{\perp}$.
Then,
\begin{align}&\E \left[\|z' - \sqrt{1 - \lambda^2}v\|_2 \right] \nonumber \\ 
&\quad= \frac 1{|S|} \sum_{v\in S} \frac 12\left(\sqrt{\|z'\|^2 + (1 - \lambda^2) - 2\sqrt{1 - \lambda^2}\alpha_v} + \sqrt{\|z'\|^2 + (1 - \lambda^2) + 2\sqrt{1 - \lambda^2}\alpha_v}\right)\label{eq:sum-of-parts}.\end{align}
We claim that this inner term is at least $\sqrt{\|z'\|^2 + (1 - \lambda^2) - \alpha_v^2}$.

Proving this can be done with some squaring. 
In particular, we can consider the following sequence of implications, starting with the square of~\pref{eq:sum-of-parts}:
\begin{align*}\frac 12(\|z'\|^2 + (1 - \lambda^2)) + \frac 12\sqrt{(\|z'\|^2 + (1 - \lambda^2))^2 - 4(1 - \lambda^2)\alpha_v^2} &\overset{?}{\ge} \|z'\|^2 + (1 - \lambda^2) - \alpha_v^2\\
\sqrt{(\|z'\|^2 + (1 - \lambda^2))^2 - 4(1 - \lambda^2)\alpha_v^2} &\overset{?}{\ge} \|z'\|^2 + (1 - \lambda^2) - 2\alpha_v^2\\
(\|z'\|^2 + (1 - \lambda^2))^2 - 4(\|z'\|^2 + (1 - \lambda^2))\alpha_v^2+4\alpha_v^4 &\overset{?}{\le} (\|z'\|^2 + (1 - \lambda^2))^2 - 4(1 - \lambda^2)\alpha_v^2  \\
4\alpha_v^4 &\overset{?}{\le} 4\|z'\|^2 \alpha_v^2.
\end{align*}
The last statement is true since $\alpha_v^2 \le \|z'\|_2^2$.
Thus, 
\begin{align*}
    \E [\|z' - \sqrt{1 - \lambda^2}v\|_2] &\ge \frac 1{|S|} \sum_{v\in S} \sqrt{\|z'\|^2 + (1 - \lambda^2) - \alpha_v^2}\\
    &\ge \frac 1{|S|} \sum_{v\in S} \left(\sqrt{\|z'\|^2 + (1 - \lambda^2)} - \frac{\alpha_v^2}{\sqrt{\|z'\|^2 + (1 - \lambda^2)}}\right)\\
    &= \sqrt{\|z'\|^2 + (1 - \lambda^2)} - \frac{\sum_{v\in S}\alpha_v^2}{|S|\sqrt{\|z'\|^2 + (1 - \lambda^2)}}\\
    &\ge \sqrt{\|z'\|^2 + (1 - \lambda^2)} - \frac{\|z\|^2}{|S|\sqrt{\|z'\|^2 + (1 - \lambda^2)}}.
\end{align*}
In an intermediate step, we used the fact that $\sqrt{x - \eps} \ge \sqrt{x} - \frac{\eps}{\sqrt{x}}$ when $x, \eps > 0$, and further we used that $u\perp v$ to see that $\sum_{v\in S} \alpha_v^2 \le \|z\|^2$.
By squaring, we see that
\begin{align*}
\E [\|z' - \sqrt{1 - \lambda^2}v\|_2]^2 &\ge \|z'\|^2 + (1 - \lambda^2) - \frac{2}{|S|} \|z\|^2\\
&= 1 + \|z\|^2 - 2\lambda \langle z, u\rangle  - \frac{2}{|S|}\|z\|^2
\end{align*}
On the other hand, our desired bound satisfies
\begin{align*}
1 - \lambda a^\ast + \frac{\lambda}{a^\ast} \|z - a^\ast u\|^2 &= 1 - \lambda \left(\lambda + \frac{2}{|S|}\right) + \frac{\lambda}{\lambda + \frac{2}{|S|}}\|z\|^2 - 2\lambda \langle z, u\rangle + \lambda\left(\lambda + \frac{2}{|S|}\right)\\
&= 1 + \frac{\lambda}{\lambda + \frac{2}{|S|}}\|z\|^2 - 2\lambda \langle z, u\rangle.
\end{align*}
Comparing these two, we need
\[\frac{\frac{2}{|S|}}{\lambda + \frac{2}{|S|}}\|z\|^2 \ge \frac{2}{|S|}\|z\|^2\]
which is true since $\lambda + \frac{2}{|S|} < 1$.
This completes the proof.
\end{proof}

Finally, we prove~\pref{clm:distance-to-final-point}.

\restateclaim{clm:distance-to-final-point}

\begin{proof}
    The proof is some algebra.
By orthonormality of the $v_i$, we can write
\[\frac{\|p_t - p_T\|_2^2}{\|p_T\|_2^2} = (\lambda^t (a^\ast)^{T-t} - \lambda^T)^2 + (1 - \lambda^2) \sum_{i=1}^{t} (\lambda^{t-i} (a^\ast)^{T-t} - \lambda^{T-i})^2 + (1-\lambda^2) \sum_{i=t+1}^{T} \lambda^{2(T-i)}.\]
It is useful to combine the last two sums as follows:
\begin{align*}
    \sum_{i=1}^{t} (\lambda^{t-i} (a^\ast)^{T-t} - \lambda^{T-i})^2 &+ \sum_{i=t+1}^{T} \lambda^{2(T-i)}\\
    &= \lambda^{2T}\sum_{i=1}^{T} \lambda^{-2i} + (a^\ast)^{2(T-t)}\lambda^{2t}\sum_{i=1}^{t}\lambda^{-2i} - 2(a^\ast)^{T-t}\lambda^{T+t}\sum_{i=1}^{t} \lambda^{-2i}.
\end{align*}
We can evaluate 
\[(1-\lambda^2)\sum_{i=1}^{x} \lambda^{-2i} = (1 - \lambda^2)\frac{\lambda^{-2(x+1)}-\lambda^{-2}}{\lambda^{-2} - 1} = \lambda^{-2x} - 1.\]
Thus, it follows that the entire sum is
\begin{align*}&\left[(\lambda^t (a^\ast)^{T-t} - \lambda^T)^2\right] + (1 - \lambda^2) \sum_{i=1}^{t} (\lambda^{t-i} (a^\ast)^{T-t} - \lambda^{T-i})^2 + (1-\lambda^2) \sum_{i=t+1}^{T} \lambda^{2(T-i)} = \\
&= \left[\lambda^{2t} (a^\ast)^{2(T-t)} + \lambda^{2T} - 2\lambda^{T+t} (a^\ast)^{T-t}\right] + 1 - \lambda^{2T} + (a^\ast)^{2(T-t)}(1 - \lambda^{2t}) - 2(a^\ast)^{T - t}(\lambda^{T-t} - \lambda^{T + t})\\
&= 1 + (a^\ast)^{2(T-t)} - 2(\lambda a^\ast)^{T-t}\end{align*}
by cancelling like terms, as desired.
\end{proof}

\end{document}